\documentclass[12pt,onecolumn]{IEEEtran}

\usepackage[top=2in, bottom=1.5in, left=1.2in, right=1.2in]{geometry}

\usepackage{amsmath}
\usepackage{amsthm}
\usepackage{amssymb}
\usepackage{amsthm}

\usepackage{graphicx}
\usepackage{color}

\usepackage{url}

\newcommand{\R}{\mathbb{R}}
\newcommand{\C}{\mathbb{C}}

\newcommand{\Q}{\mathbb{Q}}
\newcommand{\Z}{\mathbb{Z}}
\newcommand{\F}{\mathbb{F}}	

\newcommand{\CF}{C^{(4)} }

\newcommand{\CL}{C^{(l)} }
\newcommand{\DL}{D^{(l)} }
\newcommand{\CT}{C^{(2)} }
\newcommand{\DT}{D^{(2)} }

\newcommand{\diag}{\text{diag}}

\newcommand{\Th}{\vartheta}

\newtheorem{proposition}{Proposition}
\newtheorem{lemma}[proposition]{Lemma}
\newtheorem{corollary}[proposition]{Corollary}

\newtheorem{theorem}{Theorem}

\theoremstyle{remark}
\newtheorem{remark}{Remark}

\begin{document}

\title{Counterexample to the Generalized Belfiore-Sol\'{e} Secrecy Function Conjecture for $l$-modular lattices} 

\author {Anne-Maria Ernvall-Hyt\"onen and B. A. Sethuraman

\thanks{Anne-Maria Ernvall-Hyt\"onen is with Department of Mathematics and Statistics,  00014 University of Helsinki, Finland. Email: anne-maria.ernvall-hytonen@helsinki.fi}
\thanks{B.A. Sethuraman is with Department of Mathematics, California State University Northridge,  Northridge, CA 91330, USA. Email: al.sethuraman@csun.edu}
\thanks{The authors wish to thank Daniel Katz for some illuminating discussions at the start of this project and for reading portions of the paper, Jean-Claude Belfiore for some discussions on the results in this paper and on further directions, and Eric Rains for some discussions on the proof of Theorem \ref{2-mod_ratnl_eq}. Ernvall-Hyt\"onen wishes to thank the mathematics department of California State University Northridge for her visit there during which this research was initiated. The research of Ernvall-Hyt\"onen was funded by the Academy of Finland grants 138337 and 138522, while that of B.A. Sethuraman was supported by U.S. National Science Foundation grant CCF-1318260.}

}

\maketitle

\begin{abstract} We show that the secrecy function conjecture that states that the maximum of the secrecy function of an $l$-modular lattice occurs at $1/\sqrt{l}$ is false, by proving that the $4$-modular lattice $\CF = \Z \oplus \sqrt{2}\Z \oplus 2\Z$ fails to satisfy this conjecture.  We also indicate how the secrecy function must be modified in the $l$-modular case to have a more reasonable chance for it to have a maximum at $1/\sqrt{l}$, and show that the conjecture, modified with this new secrecy function, is true for various $2$-modular lattices.
\end{abstract}


\section {Introduction} \label{secn_intro}

Recall \cite{Queb} that an integral lattice $\Lambda \subset \R^n$ is said to be \textit{$l$-modular} if there exists a similarity of $\R^n$ of norm $l$, that is, an orthogonal transformation $S$ followed by a scaling of lengths by $\sqrt{l}$, such that $\sqrt{l}S(\Lambda^*) = \Lambda$. Here, $\Lambda^*$ is the dual of $\Lambda$, and $\Lambda\subset \Lambda^*$ because of integrality. It follows from elementary considerations that $l$ must necessarily be an integer and that $\Lambda$ must have determinant $l^{n/2}$. Since the determinant of $\Lambda$ is an integer, we find immediately that $n$ must be even, unless $l$ is itself a square. When $l=1$, of course, an $l$-modular lattice is known as a \textit{unimodular} lattice.

The secrecy function was introduced in \cite{OggBelSecFun} by Oggier and Belfiore, who considered the problem of wiretap code design for the Gaussian channel, using lattice-based coset coding. The function was further refined by Belfiore and Sol\'{e}  in \cite[Definition 3]{BS} to take into account the volume of the lattice $\Lambda$. It  is defined for an $l$-modular lattice $\Lambda$ (actually for any lattice) in dimension $n$  by 
 
\begin{equation} \label{old_sf}
\Xi(y) = \frac {\Theta_{\lambda \Z^n}(y)}{\Theta_{\Lambda}(y)}:= \frac {\Theta_{\lambda \Z^n}(\imath y)}{\Theta_{\Lambda}(\imath y)}.
\end{equation}
Here,  $y$ is a positive real variable, $\lambda = l^{n/4}$ is the volume of the $l$-modular lattice $\Lambda$, $\lambda \Z^n$ denotes the cubic lattice $\Z^n$ scaled to have the volume $\lambda$ (thus, each dimension of $\lambda \Z^n$ is scaled by $l^{1/4}$), and for any $\tau \in \C$ with $im(\tau) > 0$ and any lattice $L$, $\Theta_{L}(\tau)$  denotes the \textit{theta series} of $L$, that is, the series $\sum_{j=0}^{\infty} a_j e^{\imath \pi  j\tau}$, were $a_j$ is the number of vectors in $L$ of norm  (squared length)  $j$. As indicated in the equation above, when working exclusively with purely imaginary values $\imath y$ of $\tau$, we will simply write $\Theta_{L}(y)$ for $\Theta_{L}(\imath y)$.

The secrecy function was studied in detail in \cite{BS} by Belfiore and Sol\'{e}. Assuming that the noise variance $\sigma_e^2$ on  Eve's channel is much higher than the corresponding variance $\sigma_e^2$ on Bob's channel, they analyze the probability of both users making a correct decision, and determine conditions under which Eve's probability of correct decoding is minimized.  If $\Lambda_e \subset \Lambda_b$ are the lattices used in the coset-coding paradigm, they express these conditions  in terms of the theta series of $\Lambda_e$. For  a given choice of lattice $\Lambda_e$, it follows from these considerations that the value of $y$ at which the secrecy function $\Xi_{\Lambda_e}(y)$ of $\Lambda_e$  obtains its maximum yields the value of the signal-to-noise ratio in Eve's channel that causes maximum confusion to Eve, as compared to using the standard lattice $\Z^n$.
(The maximal achievable value of the secrecy function is called the \textit{secrecy gain} of the lattice $\Lambda_e$.)

Belfiore and Sol\'{e} studied the secrecy function for various lattices and conjectured in \cite{BS} that for a unimodular lattice ($l=1$), the secrecy function assumes its (global) maximum at $y=1$.  This has since been verified for a large number of lattices (see e.g.,  \cite{ErnHyt}, \cite{LinOgg}, \cite{Julie}, \cite{JulieSeth}), and it was proven in \cite{Julie} that  infinitely many unimodular lattices satisfy the conjecture, but the full conjecture is still open. In \cite{OggSoleBel}, Oggier, Sol\'{e} and Belfiore further extended this conjecture to $l$-modular lattices ($l > 1$):  they conjectured that the secrecy function of $l$-modular lattices attains its (global) maximum  at $y= 1/\sqrt{l}$ (\cite[Proposition 2, and Conjecture 1]{OggSoleBel}.  

We show in this paper that this extended conjecture is false in general.  We show that the $4$-modular lattice $\CF = \Z \oplus \sqrt{2}\Z \oplus 2\Z$  fails to satisfy the conjecture. We show that in fact that the secrecy function of $\CF$ has a global \textit{minimum} at $y = 1/\sqrt{4}$, and thus behaves contrary to what is expected by the conjecture. 

We also indicate how the conjecture must be modified to have a reasonable chance of being true:
 the numerator in the secrecy function should be replaced by a suitable power of the theta series of $\DL$, where $\DL = \Z \oplus \sqrt{l}\Z$.  We show that the modified secrecy function conjecture holds for various $2$-modular lattices, and in fact, provide a necessary and sufficient criterion for a $2$-modular lattice to satisfy  the modified conjecture.

\section {The Lattice $\CF$.} \label{CF_Secn}
In this section we show that for the $4$-modular lattice $\CF = \Z \oplus \sqrt{2}\Z \oplus 2\Z$, the secrecy function of $\CF$ defined in Equation \ref{old_sf} attains a minimum at $y = 1/2$, showing that the secrecy function conjecture is false in general.  First note that $\CF$ is indeed $4$-modular: it is easy to see that its dual is the lattice $\Z \oplus (1/\sqrt{2})\Z \oplus (1/2)\Z$, and the map $\R^3 \rightarrow \R^3$ that sends $(x,y,z)$ to $(2z, 2y, 2x)$ indeed provides an isomorphism between $\Z \oplus (1/\sqrt{2})\Z \oplus (1/2)\Z$ and $\CF$, and this map is indeed a similarity that  multiplies lengths by $2$ (and norms by $4$).

Recall the Jacobi theta functions $\Th_3(q)$, $\Th_2(q)$ and $\Th_4(q)$, where $q = e^{\imath \pi \tau}$, $im(\tau) > 0$. We will interchangeably use the notation  $\Th_3(\tau)$, $\Th_2(\tau)$ and $\Th_4(\tau)$ when thinking of these as functions of $\tau$ instead of $q$, the usage will be clear from the context. These are given by 
\begin{eqnarray} \label{prod_reps}
\vartheta_2(q) =\vartheta_2(\tau) & =& \sum_{n=-\infty}^{\infty}q^{(n+1/2)^2}=2q^{1/4}\prod_{n=1}^{\infty}(1-q^{2n})(1+q^{2n})^2\\ \nonumber
\vartheta_3(q)=\vartheta_3(\tau) &=& \sum_{n=-\infty}^{\infty}q^{n^2}=\prod_{n=1}^{\infty}(1-q^{2n})(1+q^{2n-1})^2\\ \nonumber
\vartheta_4(q) =\vartheta_4(\tau) & =& \sum_{n=-\infty}^{\infty}(-1)^n q^{n^2} =\prod_{n=1}^{\infty}(1-q^{2n})(1-q^{2n-1})^2.
\end{eqnarray}

These functions satisfy, for instance, the following formulas (\cite[page 104]{conwaysloane}):
\begin{eqnarray}\label{theta1}
\vartheta_3^4(\tau)&=\vartheta_2^4(\tau)+\vartheta_4^4(\tau)\\ \nonumber
2 \vartheta_3^2(2{\tau})&=\vartheta_3^2(\tau)+\vartheta_4^2(\tau)\\ \nonumber
2 \vartheta_2^2(2{\tau})&=\vartheta_3^2(\tau)-\vartheta_4^2(\tau).
\end{eqnarray}
(Notice that the last two equations yield $\vartheta_3^2(\tau) = \vartheta_3^2(2{\tau}) + \vartheta_2^2(2{\tau})$.)

In this paper we will be concerned with purely imaginary values of $\tau$: $\tau = \imath y$ where $y>0$.  As with theta series of lattices, we will simply write $\Th_3(y)$, $\Th_2(y)$ and $\Th_4(y)$ for $\Th_3(\imath y)$, $\Th_2(\imath y)$ and $\Th_4(\imath y)$. The Jacobi theta functions $\vartheta_2$, $\vartheta_3$ and $\vartheta_4$ are useful in representing the theta functions of various lattices. 
A thorough introduction to the theory of these functions can be found in \cite[Chap. 10]{steinandshakarchi}, in terms of the ``master'' theta function $\Theta(z|\tau) = \sum_{n=-\infty}^{\infty}e^{2\pi\imath n z + \pi\imath n^2 \tau}$. (We may write our functions $\Th_2$, $\Th_3$, $\Th_4$ in terms of $\Theta$ as $\Th_2(\tau) = e^{\imath\pi\tau/4} \Theta(\dfrac{\tau}{2}|\tau)$, $\Th_3(\tau) = \Theta(0|\tau)$, and $\Th_4(\tau) = \Theta(\dfrac{1}{2}|\tau)$--see \cite[page 102]{conwaysloane} for instance, but note the slight difference in the definitions of $\Theta$ in \cite{steinandshakarchi} and \cite{conwaysloane}.)


Note that the theta series of $\CF$ (for $\tau = \imath y$, $y>0$) is given by $\Th_3(y)\Th_3(2y)\Th_3(4y)$, and the theta series of {$(\sqrt{2}\Z)^3$ is given by $\Th_3(2y)^3$}.  We 
find it convenient to work with the reciprocal of the secrecy function:
\begin{equation} \label{CF_SF_eqn}
1/\Xi_{\CF}(y) = \frac{\Th_3(y)\Th_3(2y)\Th_3(4y)}{\Th_3^3(2y)} = \frac{\Th_3(y)\Th_3(4y)}{\Th_3^2(2y)}.
\end{equation}
We find it convenient as well to put $z=2y$. Thus, to show that
the secrecy function of $\CF$ defined in Equation \ref{old_sf} attains a minimum at $y = 1/2$, we need to show that the modified function
\begin{equation} \label{C4_f_defn}
f(y) = \frac{\Th_3(y/2)\Th_3(2y)}{\Th_3^2(y)}
\end{equation}
(where by abuse of notation we have retained the symbol $y$ for the new variable $z$) has a  maximum at $y=1$.

We now invoke  results connecting theta functions at the purely imaginary values $\tau = \imath y$ ($y>0$)  and $\tau/2$ (i.e., at $q = e^{-\pi y}$ and $\sqrt{q}$) from \cite{B2}; a summary of what we need is in \cite[Section 4.6, Page 137]{B2}. We build on the notation ``$k$'' and ``$l$'' of \cite{B2} and write more specifically $k(q)$, $k'(q)$, $l(q)$, and $l'(q)$ for the objects:
\begin{eqnarray} \label{kl_def}
k(q) &=& \frac{\Th_2^2(q)}{\Th_3^2(q)}\\
k'(q) &=& \sqrt{1-k^2(q)} = \frac{\Th_4^2(q)}{\Th_3^2(q)} \nonumber \\
l(q) &=& k(\sqrt{q}) = \frac{\Th_2^2(\sqrt{q})}{\Th_3^2(\sqrt{q})} \nonumber\\
l'(q) &=& k'(\sqrt{q}) = \frac{\Th_4^2(\sqrt{q})}{\Th_3^2(\sqrt{q})} \nonumber
\end{eqnarray}

(The expression for $k'(q)$ arises from the first of Equations \ref{theta1} above.)
Finally, we write 
\begin{equation} 
M_2(q) = \frac{\Th_3^2(q)}{\Th_3^2(\sqrt{q})}.
\end{equation}

As described in \cite{B2}, $M_2(q)$  can be written in terms of $k(q)$, $k'(q)$, $l(q)$, and $l'(q)$, and further, $k(q)$ and $l(q)$ are connected by a ``modular equation.'' We have the relations (\cite[Section 4.6, Page 137]{B2} (these can also be directly derived from the properties of theta functions in Equations \ref{theta1})  
\begin{equation} \label{M2_kl}
M_2(q) = \frac{1}{1+k(q)} = \frac{1+l'(q)}{2},
\end{equation} and
\begin{eqnarray}  \label{mod_eq}
l(q) &=& \frac{ 2\sqrt{k(q)}} {1+k(q)}\\
k(q) &=& \frac{1-l'(q)}{1+l'(q)} \nonumber
\end{eqnarray}

Since $f(y) = \dfrac{M_2(q^2)}{M_2(q)}$,  Equations \ref{M2_kl} shows that
\begin{equation} \label{expressn_to_maximize}
f(y)  = \frac{1+k(q)}{1+k(q^2)} = \frac{(1+k(q))(1+l'(q^2))}{2} = \frac{(1+k(q))(1+k'(q))}{2}.
\end{equation}

Thus, we need to maximize $(1+k(q))(1+k'(q))$ where $k(q)^2 + k'(q)^2=1$. 
Putting $k(q)= \cos(\alpha) = \dfrac{1-t^2}{1+t^2}$ and $k'(q)=\sin(\alpha) = \dfrac{2t}{1+t^2}$, where $t = \tan(\alpha/2)$, we find need to determine the extrema of
\begin{equation}
f(t) = \frac{(1+t)^2}{(1+t^2)^2}.
\end{equation}
Now $0 <k(q) <1 $ and $0 < k'(q) <1$ by definition of $k(q)$, $k'(q)$ and the relation $k(q)^2 + k'(q)^2=1$. Thus, $0 < \alpha < \pi/2$, so $0 < \alpha/2 < \pi/4$.  It follows that $0 < t < 1$.  Calculus now shows that 
that $t=\sqrt{2}-1$ is the unique (and hence global) maximum of $f(t)$ in the region $0 < t < 1$.

Corresponding to $t=\sqrt{2}-1$, we find $\alpha/2 = \pi/8$, i.e., $\alpha=\pi/4$.  Thus, $q$ is such that $k(q) = k'(q)$, i.e., $\Th_2(y) = \Th_4(y)$. This occurs precisely at $y=1$ (see for instance \cite[Proof of Lemma 1]{JulieSeth}, or \cite[Exercise 4, Section 2.3]{B2} along with \cite[Exercise 8b, Section 3.1]{B2}).  Further, we see that $f(y)$  considered as a function of $y$ has the same increase/decrease behavior on either side of $y=1$ as $f(t)$ does on either side of $t=\sqrt{2}-1$ when considered as a function of $t$: The map $y \mapsto k(e^{-\pi y})$ is a monotonically decreasing map (\cite[Equation 2.3.9, Page 42] {B2}, this also follows from Lemma \ref{kprime_increasing} ahead, and the fact that $k^2 + k'^2 = 1$), while the map $k(e^{-\pi y}) = \cos(\alpha) \mapsto t = \tan(\alpha/2)$ is also monotonically decreasing. The chain rule now shows that $df/dy$ and $df/dt$ have the same sign. It follows that $f(y)$ increases for $0 < y < 1$ and decreases for $1<y<\infty$; correspondingly, since $1/\Xi_{\CF}(y) = f(2y)$, we find $\Xi_{\CF}$ \textit{decreases} for $0 < y < 1/2$ and \textit{increases} for $1/2 < y < \infty$.

Thus, $\CF$ violates the conjecture.

\begin{remark}\label{C4_graph} The graph of the secrecy function of $\CF$ may be computed (approximately), using Mathematica$^\circledR$. The graph is shown in Figure \ref{counter4}, and verifies our analysis above.
\begin{figure}[h]
    \centering
    \includegraphics[width=0.8\textwidth]{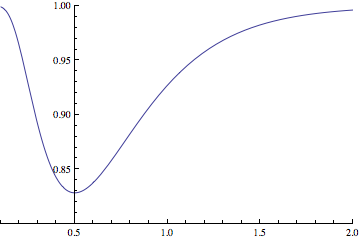}
    \caption{\small{Graph of secrecy function of lattice $\CF$. Notice that according to the original conjecture, the function should have its maximum at $x=\frac{1}{2}$, but it has a minimum.}}
    \label{counter4}
\end{figure}

\end{remark}

\section{Modified Secrecy Function} \label{secn_mod_sf}
The current definition of the secrecy function compares the theta series of an $l$-modular lattice in $\R^n$ to the theta series of the (scaled) unimodular lattice $\Z^n$. A more natural definition would be one that compared likes with likes: that compared the theta series of an $l$-modular lattice to that of another reference $l$-modular lattice, scaled suitably to match volumes.  

The simplest $l$-modular lattice is $\DL = \Z \oplus \sqrt{l}\Z$ (when $l=1$, we take $\DL = \Z$). Note that $\DL$ can be proved to be $l$-modular exactly like the lattice $\CF$ in Section \ref{CF_Secn}---the dual is the lattice $\Z \oplus (1/\sqrt{l})\Z$, and the required map on $\R^2$ is the one that takes $(x,y)$ to $(\sqrt{l}y, \sqrt{l}x)$. Accordingly, we write $n = k \dim(\DL)$ ($ = 2k$ for $l>1$), and for an  $l$-modular lattice $\Lambda$ in $\R^n$, we define the \textit{$l$-modular secrecy function} $\Xi_l(y)$ (or $\Xi_{l,\Lambda}(y)$ if the lattice $\Lambda$ needs to be emphasized), by
\begin{equation} \label{modular_sf}
\Xi_l(y) =\Xi_{l,\Lambda}(y) := \frac{\Theta_{\DL}(y)^k} {\Theta_\Lambda(y)}, \quad y>0.
\end{equation}
(Note that when $l=1$, $k=n$, $\DL=\Z$, and this definition reduces to the earlier definition of the secrecy function of a unimodular lattice.)

When $l$ is not a square, $n$ must necessarily be even, as we have noted in Section \ref{secn_intro}.
When $l$ is a square, $n$  need not be even, as the example of $\CF$ attests. In such cases, the definition above of the secrecy function involves a square root of the theta series of $\DL$. (Of course, we are scaling up the theta series, not the lattice!)

It is reasonable now to modify the original conjecture and make the following \textit{$l$-modular secrecy function conjecture:} that for all $l$-modular lattices, the $l$-modular secrecy function attains its (global) maximum at $1/\sqrt{l}$.  We show in the next section that this new conjecture holds for various $2$-modular lattices in small dimension.  But we can see immediately that it holds for $\CF$ as follows:

\begin{equation}
\Xi_l(y) =\Xi_{l,\CF}(y) = \frac{(\Th_3(y)\Th_3(4y))^{3/2}} {\Th_3(y)\Th_3(2y)\Th_3(4y)} = \frac {(\Th_3(y)\Th_3(4y))^{1/2}} {\Th_3(2y)}.
\end{equation}
But we have already seen above in Section \ref{CF_Secn} that $\Xi_l(y)^2 =  \dfrac {\Th_3(y)\Th_3(4y)} {\Th_3^2(2y)}$ has a global maximum at $y = 1/2$, so $\Xi_l(y)$ also has a global maximum at $y=1/2$. Thus, our modified conjecture is true for $\CF$.

\begin{remark} The $l$-modular secrecy function exhibits ``multiplicative symmetry'' about the point $1/\sqrt{l}$, that is, $\Xi_l(a) = \Xi_l(b)$ when $ab=1/l$. The proof is the same as that for the originally defined secrecy function, \cite[Prop. 2]{OggSoleBel}.
\end{remark}

\section {$2$-modular lattices} \label{secn_2_mod}

In the following, we will show that the $l$-modular secrecy function conjecture stated above holds for all the $2$-modular lattices considered in \cite {LinOggSole}.  The starting point is the following description of the theta series of such a lattice:
\begin{theorem} \label{theta_poly_2mod}
The theta series of a $2$-modular lattice $\Lambda$ in dimension $n= 2k$ is a polynomial
\begin{equation} \label{theta_poly_eq}
\Theta_{\Lambda}(y) = f_1(y)^k\left( \sum_{i=0}^{\lfloor k/2 \rfloor} a_i f_2(y)^i\right) =  \sum_{i=0}^{\lfloor k/2 \rfloor} a_i f_1^{k-2i} \Delta_4(y)^i, 
\end{equation}
where $f_1(y) = \Theta_{\CT}(y)$,  $f_2(y) = \dfrac{\Th_2^2(2y) \Th_4^2(y)}{4 \Th_3^2(y) \Th_3^2(2y)}$, and $\Delta_4 = f_1^2f_2$.
\end{theorem}

This theorem follows from a more general theorem of Rains and Sloane (\cite[Theorem 9, Corollary 3]{RS}. Although it is only applied to odd $2$-modular lattices in \cite[Eqns 29, 30]{LinOggSole}, the theorem above actually holds for any $2$-modular lattice.  We can see this as follows: the theorem and corollary referred to in \cite{RS} apply to strongly $l$-modular lattices that are rationally equivalent to $(\CL)^k$ (where $n=2k$).  The definition of strongly modularity in \cite{RS} (see the discussion in that paper, following Theorem 6) shows that when $l$ is prime, any $l$-modular lattice is automatically strongly $l$-modular. We thus need the following theorem to enable us to apply the results of \cite[Theorem 9, Corollary 3]{RS} to any $2$-modular lattice: 

\begin{theorem} \label{2-mod_ratnl_eq}
 All $2$-modular lattices in $\R^n$ ($n=2k$) are rationally equivalent to $(\CT)^k$.
\end{theorem}\begin{proof} This is well known, and falls out easily from the classification theorem for quadratic forms over $\Q$. For lack of a specific reference, we sketch the proof in Appendix \ref{App_2-mod}. \end{proof}

For us, since $2$ is prime, $\CT$ is the same as $\DT$,  we  find
\begin{equation} \label{Xi_odd_2_mod}
\Xi_{2,\Lambda}(y) = \Xi_2(y) = \left(\sum_{i=0}^{\lfloor k/2 \rfloor} a_i f_2(y)^i\right)^{-1}.
\end{equation}
We will study the general behavior  of such a polynomial function of $f_2$ and then apply our results to the specific theta series computed in \cite{LinOggSole}.

We have, using Equations \ref{theta1}: 
\begin{equation} \label{f_2_desc}
f_2(y)=\frac{\vartheta_2^2(2y)\vartheta_4^2(y)}{4\vartheta_3^2(y)\vartheta_3^2(2y)}=\frac{(\vartheta_3^2(y)-\vartheta_4^2(y))\vartheta_4^2(y)}{4(\vartheta_3^2(y)+\vartheta_4^2(y))\vartheta_3^2(y)}=\frac{(1-\alpha)\alpha}{4(1+\alpha)},
\end{equation}
where $\alpha=\alpha(y) = \dfrac{\vartheta_4^2}{\vartheta_3^2}(y)$.

\begin{lemma} \label{kprime_increasing}
 The function $\dfrac{\vartheta_4}{\vartheta_3}(y )$ is strictly increasing  for (positive) real y, and as $y\rightarrow 0$, the function approaches $0$, and as $y\rightarrow \infty$, the function approaches $1$.
\end{lemma}
\begin{proof} A formal proof that takes care of intricacies of infinite products and interchanges of limits is in Appendix \ref{App_T4/T3}. The intuition is as follows: Using the product representations of $\vartheta_4$ and $\vartheta_3$ in Equations \ref{prod_reps}, we have
\begin{eqnarray*}
\frac{\vartheta_4}{\vartheta_3}(y) &=& \frac{\prod_{m=1}^{\infty}(1-q^{2m})(1-q^{2m-1})^2}{\prod_{m=1}^{\infty}(1-q^{2m})(1+q^{2m-1})^2}\\ 
&=&\prod_{m=1}^{\infty}\left(\frac{1-q^{2m-1}}{1+q^{2m-1}}\right)^2=\prod_{m=1}^{\infty}\left(\frac{2}{1+q^{2m-1}}-1\right)^2 \\ 
&=& \prod_{m=1}^{\infty}\left(\frac{2}{1+q^{2m-1}}-1\right)^2. 
\end{eqnarray*}

Now, as $y$ increases, $q$ decreases, and hence, $\left(\frac{2}{1+q^{2m-1}}-1\right)$ increases. This shows that the function is increasing. Furthermore, as $y\rightarrow 0$, $q^{2m-1}\rightarrow 1$, and $\left(\frac{2}{1+q^{2m-1}}-1\right)\rightarrow 0$. As $y\rightarrow \infty$, $q^{2m-1}\rightarrow 0$, and $\left(\frac{2}{1+q^{2m-1}}-1\right)\rightarrow 1$.

\end{proof}

\begin{lemma} \label{max_f}
The function
\[
f(x)=\frac{(1-x)x}{(1+x)}
\]
has a unique maximum in the open interval $(0,1)$, and this maximum is met at the point $x = \sqrt{2}-1$.
\end{lemma} 
\begin{proof} This is straightforward.
\end{proof}

\begin{remark} \label{rem:val_f2_at_peak}
The value of $f_2$ when $\alpha = \sqrt{2}-1$ is $\dfrac{(1-(\sqrt{2}-1))(\sqrt{2}-1)}{4(1+\sqrt{2}-1)} \approx 0.0429$. We will denote this value by $\beta$ in what follows. 
\end{remark}
\begin{lemma} \label{lem:val_alpha_at_peak}
The quantity $\dfrac{\vartheta_4^2}{\vartheta_3^2}(y)$ takes on the value $\sqrt{2}-1$ precisely when $y={1}/{\sqrt{2}}$.
\end{lemma}

\begin{proof} This is in Appendix \ref{app:alpha_etc.} \end{proof}

 We now use the previous results to prove the following: 
\begin{proposition}
A necessary and sufficient condition for $\Xi_2(y)$ to have a global maximum at $y = 1/\sqrt{2}$ is that the polynomial $\left(\Xi_2(f_2)\right)^{-1} = \left(\sum_{i=0}^{\lfloor k/2 \rfloor} a_i f_2(y)^i\right)$ in the variable $f_2$  (Equation \ref{Xi_odd_2_mod}), restricted to the domain $0 < f_2 \le \beta$ where $\beta$ as in Remark \ref{rem:val_f2_at_peak} above, have a global minimum at $f_2 = \beta$.
\end{proposition}
\begin{proof}By Equation \ref{f_2_desc}, $f_2(y) = \dfrac{(1-\alpha)\alpha}{4(1+\alpha)}$, 
where $\alpha=\alpha(y)=\dfrac{\vartheta_4^2}{\vartheta_3^2}(y)$, so by Lemma \ref{max_f}, $f_2(\alpha)$ has a unique maximum when $\alpha = \sqrt{2}-1$.  By Remark \ref{rem:val_f2_at_peak} this maximum is $\beta$. Moreover, by Lemma \ref{lem:val_alpha_at_peak}, $\alpha = \sqrt{2}-1$ precisely when $y = 1/\sqrt{2}$. Thus, for other values of $y$, $f_2(y) < \beta$, and of course, $f_2(y) > 0$ by the definition of $f_2$ and by the fact that $\alpha\in(0,1)$. We thus find that as $y$ ranges in $(0,\infty)$, $f_2(y)$ ranges in $(0,\beta]$, and $f_2(y) = \beta$ precisely when $y=1/\sqrt{2}$. 
It is now clear that $\Xi_2(y)$, with $0 < y < \infty$, attains its  global maximum when $y = 1/\sqrt{2}$ if and only if $\left(\Xi_2(y)\right)^{-1} $, with $0 < y < \infty$, attains its  global minimum when $y = 1/\sqrt{2}$ if and only if
$\left(\Xi_2(f_2)\right)^{-1} $, with $f_2\in(0, \beta]$, attains its global minimum at $f_2 = \beta$.

\end{proof}
\begin{corollary}\label{decr_is_enough}
If the polynomial $\left(\Xi_2(f_2)\right)^{-1}$ is decreasing in $(0,\beta]$, then $\Xi_2$ has a global maximum at $y=1/\sqrt{2}$.
\end{corollary}

We now consider the odd $2$-modular lattices in \cite[Table 2]{LinOggSole}.  The authors have computed their theta series in terms of $f_2$ and $\Delta_4 = f_1^2 f_2$. Factoring $f_1^{n/2}$ from these series (where $n$ is the ambient dimension), we have the following table, where the third column contains the derivative of the polynomial $\left(\Xi_2(f_2)\right)^{-1}$, and the fourth column checks that this derivative is negative in $(0,\beta]$, i.e, (Corollary \ref{decr_is_enough}) that $\left(\Xi_2(f_2)\right)^{-1}$ is decreasing in $(0,\beta]$:
\vspace {0.2 in}
\begin{center} 
\begin{tabular}{|c|l|l|c|} 
\hline
Dim & $\left(\Xi_2(f_2)\right)^{-1}$ & $d/df_2 \left(\Xi_2(f_2)\right)^{-1}$ & Neg in $(0,\beta]$? \\
\hline
8 & $1-8f_2$ & $-8$ & Yes\\ 
\hline
12 & $1-12 f_2$ & $-12$ & Yes\\
\hline
16 & $1-16 f_2$ & $-16$ & Yes\\
\hline
18 & $1-18 f_2 + 18f_2^2$ & $-18 + 36f_2$ & Yes\\
\hline
20 & $1-20 f_2 + 40 f_2^2$ & $-20 + 80f_2$ & Yes\\
\hline
22 &  \begin{tabular}{l}$1-22 f_2 + 66 f_2^2$\\  $- 4f_2^3$ \end{tabular} & $-22 + 132f_2 -12 f_2^2$& Yes\\
\hline
24 &  \begin{tabular}{l}$1-24 f_2 + 96 f_2^2$\\ $- 28f_2^3$ \end{tabular} & $-24 + 192f_2 -84 f_2^2$& Yes\\
\hline
26 &  \begin{tabular}{l}$1-26 f_2 + 130 f_2^2$ \\$ - 80f_2^3$ \end{tabular} & $-26 + 260 f_2 -240 f_2^2$& Yes\\
\hline
28 & \begin{tabular}{l}$1-28  f_2 + 168 f_2^2 $  \\ $- 176f_2^3+ 32 f_2^4 $\end{tabular} & \begin{tabular}{l}$-28 + 336 f_2 -528 f_2^2$ \\ $+ 128 f_2^3$\end{tabular}& Yes\\
\hline
30 & \begin{tabular}{l}$1-30  f_2 + 210 f_2^2 $  \\ $- 282 f_2^3+ 112 f_2^4 $\end{tabular} & \begin{tabular}{l}$-30 + 420 f_2 -846 f_2^2$ \\ $+ 448 f_2^3$\end{tabular}& Yes\\
\hline

\end{tabular}
\end{center}

\vspace {0.2 in}
Clearly, the modified conjecture holds for these lattices.

For illustration, we graph the $l$-modular secrecy function for the odd $2$-modular lattice in dimension $22$ considered above in Figure \ref{22peak}.

\begin{figure}[h]
  \centering
    \includegraphics[width=0.8\textwidth]{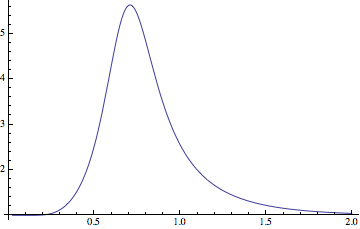}
    \label{22peak}
    \caption{Graph of the $l$-modular secrecy function of the odd $2$-modular $22$-dimensional  lattice considered in \cite[Table 2]{LinOggSole}. It has its maximum at $x=\frac{1}{\sqrt{2}}$. }
\end{figure}

We turn our attention now to the even $2$-modular lattices considered in \cite[Table 1]{LinOggSole}. There are three of them: $D_4$, $BW_{16}$, $HS_{20}$. There, their theta series have been developed in terms of two functions: the theta series of $D_4$ itself (this is a tautological statement for $D_4$ of course!), and $\Delta_{16}$:

\[
\Theta_{BW_{16}}=\Theta_{D_4}^4-96\Delta_{16} \quad \textrm{and}\quad \Theta_{HS_{20}}=\Theta_{D_4}^5-120\Theta_{D_4}\Delta_{16}.
\]
By Theorem \ref{theta_poly_2mod}, these theta series can be also expressed as polynomials in $\Theta_{C^{(2)}}$ and $\Delta_4$. By comparing coefficients, we have
\begin{align*}
\Theta_{D_4}&=\Theta_{C^{(2)}}^2-4\Delta_4\\
\Theta_{BW_{16}}&=\Theta_{C^{(2)}}^8-16\Delta_4\Theta_{C^{(2)}}^6-256\Delta_4^3\Theta_{C^{(2)}}^2+256\Delta_4^4\\
\Theta_{HS_{20}}&=\Theta_{C^{(2)}}^{10}-20\Theta_{C^{(2)}}^8\Delta_4+40\Theta_{C^{(2)}}^6\Delta_4^2
-160\Theta_{C^{(2)}}^4\Delta_4^3+1280\Theta_{C^{(2)}}^2\Delta_4^4-1024\Delta_4^5.\end{align*}
Since $\Theta_{C^{(2)}}=f_1$ and $\Delta_4=f_1^2f_2$, the $2$-modular secrecy functions $\Xi_2$ are 
\begin{align*}
\Xi_{2,D_4}&=(1-4f_2)^{-1}\\
\Xi_{2,BW_{16}}&=(1-16f_2-256f_2^3+256f_2^4)^{-1}\\
\Xi_{2,HS_{20}}&=(1-20f_2+40f_2^2-160f_2^3+1280f_2^4-1024f_2^5)^{-1}.\end{align*}

The function $(1-4f_2)^{-1}$ is clearly increasing in the range of $f_2$. As for $\Xi_{2,BW_{16}}$, the derivative of the denominator is
\[
-16-768f_2^2+1024f_2^3,
\]
which has its only real zero at $f_2\approx 0.78$, and therefore, the denominator is decreasing and the function increasing in $[0, \beta]$.
Finally, the derivative of the denominator of the $2$-modular secrecy function of $HS_{20}$ is
\[
-20+80f_2-480f_2^2+5120f_2^3-5120f_2^4.
\]
The first positive real zero is at $f_2\approx 0.17$, and therefore the denominator is decreasing and the function increasing in $[0, \beta]$.

\appendices

\section {$2$-modular lattices are rationally equivalent to direct sum of copies of $\CT$.} \label{App_2-mod} 

We sketch here the proof of Theorem \ref{2-mod_ratnl_eq}.  We assume basic familiarity with quadratic forms. The proof invokes the local-global theory of quadratic forms over number fields; we only sketch the outlines of the theory and only provide as much detail as would enable one to construct the full proof for oneself. An excellent reference is \cite{Serre}. A very readable account is also  in \cite{Pall}.

We recall first the setup behind rational equivalence: Any symmetric $n\times n$ matrix $A$ with entries in $\Q$ (such as the Gram matrix $G_L$ of an integral lattice $L$ in $\R^n$, whose entries are even in $\Z$) determines a quadratic form $q$ on $\Q^n$ in the standard way: if $e_i$ are the standard basis vectors, then $q(x_1e_1 + \cdots + x_ne_n) = v^t A v$, where $v = (x_1,\dots,x_n)^t$; here  the superscript $t$ stands for transpose.  Conversely, given a quadratic form  $q$ on $\Q^n$, we obtain a symmetric $n\times n$ matrix $A$  with rational entries, with $(i,j)$ entry given by $(q(e_i +e_j) - q(e_i) - q(e_j))/2$, .  We say two quadratic forms $q_1$ and $q_2$ on $\Q^n$ are equivalent over $\Q$, or \textit{rationally equivalent}, if there exists an invertible $n\times n$ matrix with rational entries $S$ such that $S^tA_1S = A_2$, where $A_i$ is the symmetric matrix associated with $q_i$ as above.  Alternatively, two such quadratic forms are rationally equivalent if one!
  can be obtained from the other by a linear change of variables defined over $\Q$.  (These definitions extend in the obvious way to quadratic forms over any field of characteristic different from $2$.) We apply these considerations to lattices: two integral lattices $L_1$ and $L_2$ are said to be rationally equivalent if their associated quadratic forms are rationally equivalent, or equivalently,  if the Gram matrices $G_1$ and $G_2$ of the two lattices are related by $S^t G_1 S = G_2$ for some invertible $n\times n$ matrix $S$ with rational entries.

There is a well-established theory that determines when two quadratic forms defined over $\Q$ are equivalent.  By this theory, the rational equivalence class of a quadratic form on $\Q^n$  which is non degenerate, that is, the determinant of the associated symmetric matrix is nonzero, is determined by the following objects: the \textit{discriminant}, the \textit{signature}, and the \textit{Hasse-Witt invariant} at each (integer) prime $p$. The first two are easy to describe. The \textit{discriminant} of a non degenerate quadratic form defined over $\Q$ is just the class of the determinant of the associated symmetric matrix in $\Q^*/{\Q^*}^2$. As for the signature recall first that given any symmetric $n\times n$ matrix  $A$ with entries in a field $k$ of characteristic different from $2$, there exists a nonsingular $n\times n$ matrix $S$ such that $S^tAS$ is diagonal. The \textit{signature} of a non degenerate quadratic form on $\Q^n$ is 
 just the number of positive entries minus the number of negative entries in any diagonal representation of the quadratic form, thought of as a quadratic form on $\R^n$. (The definition is independent of which diagonal representation is used.)
In our situation, note that the quadratic forms arising from the $2$-modular lattice $L$ and from $(\CT)^k$ are both positive definite, since they yield lengths of vectors in Euclidean space. It follows that all diagonalizations of either quadratic form must consist only of positive elements along the diagonal. Thus, the signature is the same for both lattices. Further, both lattices clearly have the same determinant for their associated quadratic form, namely $2^k$. Thus, to prove the rational equivalence of $L$ and $(\CT)^k$, we only need to consider their Hasse-Witt invariants, and to show that their Hasse-Witt invariants are the same at each prime $p$.

In fact, the Hasse-Witt invariant is defined not only for each integer prime $p$, but also, for $\R$. (It is traditional to think of $\R$ as the completion of $\Q$ at the ``infinite prime.'') In what follows, $v$ will denote either an integer prime $p$ or $\infty$, and $\Q_v$ will accordingly denoted either the field $\Q_p$ of $p$-adic rationals (when $v=p$) or $\R$ (when $v=\infty$). 
Given a non degenerate quadratic form $q$ over the field $\Q_v$, one first takes a diagonal representation $\diag(a_1,\dots, a_n)$ of the associated symmetric matrix.  The \textit{Hasse-Witt invariant} $\epsilon_v(q)$ is defined to be the product of the \textit{Hilbert symbols} $(a_i,a_j)_v$ over all $1 \le i < j \le n)$.  (The definition is independent of which diagonal representation is used.) In turn, given $a$ and $b$ in $\Q_v^*$, the \textit{Hilbert symbol} $(a,b)_v$ is defined to be $1$ if the equation $z^2 - ax^2 -by^2$ has a solution $(x,y,z)\neq (0,0,0)$ in $\Q_v$, and $-1$ otherwise.  A few relevant facts about the Hasse-Witt invariant and the Hilbert symbol are the following: 
\begin{enumerate}
\item \label{mostly_1} The Hasse-Witt invariant  of a quadratic form $q$ defined over $\Q^n$ is $1$ at all but at most a finite number of primes $v$.  (Here, for each prime $v$, we first view $q$ as a quadratic form over $\Q_v^n$ and then calculate $\epsilon_v(q)$.)
\item \label{recip} For a quadratic form $q$ defined over $\Q^n$, the   product over all primes $v$ of $\epsilon_v(q)$ is $1$. 
\item \label{compute_HS} For an odd (integer) prime $p$, given $a$ and $b$ in $\Q_p^*$, the Hilbert symbol $(a,b)_p$ is defined as follows: we first write $a = p^{\alpha} u$ and $b=p^\beta v$, where $u$ and $v$ are units of $\Z_p$. Then
\begin{equation} \label{HS_odd_p}
(a,b) = (-1)^{\alpha \beta (p-1)/2} \left(\dfrac{u}{p}\right)^\beta \left(\dfrac{v}{p}\right)^\alpha,
\end{equation}
where $\left(\dfrac{u}{p}\right)$ is the \textit{Legendre symbol} defined to be $1$ if the class of $u$ in $\F_p$ is a square and $-1$ otherwise.
\end{enumerate}

It follows from the characterization above that if $p$ is odd and $a$ and $b$ are themselves \textit{units} in $\Q_p$ (by units in $\Q_p$ we mean that $\alpha$ and $\beta$ above are both zero, so these are the units of $\Z_p$), then the Hilbert symbol $(a,b)_p$ is $1$. Also, the Hilbert Symbol $(a,b)_{\infty}$ is $1$ whenever both $a$ and $b$ are positive, since $z^2=ax^2+by^2$ will clearly have a nontrivial solution, e.g., $(1,0,\sqrt{a})$, or $(0,1,\sqrt{b})$. (In fact, it is enough that  just one of $a$ or $b$  is positive.) 
Now apply these considerations to the lattice $(\CT)^k$: the associated quadratic form $q_{(\CT)^k}$ is already diagonal, with $k$ $1$s and $k$ $2$s along the diagonal.  For any odd prime $p$, $1$ and $2$ are both units, and therefore, $\epsilon_p (q_{(\CT)^k}) = 1$.  It is clear too that  $\epsilon_{\infty} (q_{(\CT)^k}) = 1$ since $1$ and $2$ are positive. It follows from (\ref{recip}) above that $\epsilon_2 (q_{(\CT)^k}) = 1$ as well.

We now consider $\epsilon_v(q_L)$ at each integer prime and at infinity, where $q_L$ is the quadratic form associated to the $2$-regular lattice $L$. 
As already noted, since $q_L$ is positive definite, any diagonalization over $\R$ must consist of all positive numbers along the diagonal.  Thus, $\epsilon_{\infty}(q_L)=1$.
It is enough now to show that for any odd prime $p$, $\epsilon_p(q_L) = 1$, for then, by (\ref{recip}) above, $\epsilon_2(q_L)$ will be $  1$ as well. The key is the following proposition that describes a diagonalization.  Recall that $\Z_{ ( p ) }$ denotes the localization of $\Z$ at $p$, that is, the ring of all reduced fractions $a/b$ such that $p$ does not divide $b$;  $\Z_{ ( p ) }$ is a unique factorization domain with a single prime, namely $p$, and the reduced fraction $a/b$ above of $\Z_{ ( p ) }$  is divisible by $p$ precisely when $a$ is divisible by $p$. 
Under the embedding $\Q \mapsto \Q_p$, $\Z_{ ( p ) }$  goes to the $p$-adic integers $\Z_p$, and the elements of $\Z_{ ( p ) }$ 
not divisible by $p$ live naturally as units in the $p$-adic integers $\Z_p$.

\begin{proposition} \label{diag_q_odd_prime} Suppose that $A$ is a symmetric matrix in $M_n(\Q)$ and suppose that $p$ does not divide the determinant of $A$, where $p$ is an odd prime.  Then, there exists an  $n\times n$  matrix $S$ with entries in $\Z_{ ( p ) }$ of determinant $\pm1$ such that $S^t A S = \diag(a_1,\dots, a_n)$, where the numerators and denominators of each $a_i$, when written as a reduced fraction, is not divisible by $p$.
\end{proposition}

(For the full statement of the proposition, see \cite[Lemma 5.1]{Pall}.)

\begin{proof} We sketch the proof here. Since $p$ does not divide the determinant, some entry $a_{i,j}$ of $A$ must be prime to $p$.  First suppose that some $a_{i,i}$ is prime to $p$.  Then, we swap the basis vectors $e_1$ and $e_i$, a transformation of determinant $- 1$, to ensure that $a_{1,1}$ is prime to $p$. If all $a_{i,i}$ are divisible by $p$, some $a_{i,j}$ with $i\neq j$ must be prime to $p$. We consider $(e_i + e_j)^t A (e_i + e_j)$: this is $a_{i,i} + a_{j,j} + 2 a_{i,j}$.  Since each $a_{i,i}$ and $a_{j,j}$ are divisible by $p $ and since $p$ is odd and $a_{i,j}$ is prime to $p$, we find  $(e_i + e_j)^t A (e_i + e_j)$ is prime to $p$.  Thus, the transformation $e_1 \mapsto (e_i+e_j)$, $e_i \mapsto e_1$ is of determinant $- 1$, and ensures that $a_{1,1}$ is prime to $p$.  Thus by a change of basis with determinant $- 1$, we can ensure that $a_{1,1}$ is prime to $p$.  We now write $A$ in the block form
$$
A = 
\left(
\begin{array}{cc}
 a_{1,1} & B      \\
 B^t & C     
\end{array}
\right),
$$ and take $S$ to be the matrix
$$
\left(
\begin{array}{cc}
 1 & -a_{1,1}^{-1}B      \\
 0 & I_{n-1}     
\end{array}
\right)
$$ to find
$$
S^t A S = 
\left(
\begin{array}{cc}
 a_{1,1} & 0     \\
 0 & C  - a_{1,1}^{-1}B^tB  
\end{array}
\right).
$$ (Notice that $S$ has determinant $1$.) We now proceed by induction, working in $\Z_{ ( p ) }$, noting that the product of the various basis-change matrices at each stage has determinant $\pm 1$. .

\end{proof}

Since the determinant of the matrix associated to $q_L$ is $2^k$, we may apply
this result to $q_L$. For a given odd prime $p$, take a diagonal representation $\diag(a_1, \dots, a_n)$ of $q_L$ over $\Q_p$ as furnished by the proposition.  Each $a_i$ is nonzero element of $\Z_{ ( p ) }$ not divisible by $p$, and is therefore a unit in $\Q_p$.  Thus, by (\ref{HS_odd_p}) above, the Hasse-Witt invariant $\epsilon_p(q_L)$ is $1$. By (\ref{recip}) above, $\epsilon_2(q_L)$ is also $1$.

Since $L$ and $(\CT)^k$ have the same Hasse-Witt invariant at every prime in addition to having the same signature and discriminant, they are indeed rationally equivalent as claimed.

\section {Proof of Lemma \ref {kprime_increasing}} \label {App_T4/T3}
\begin{proof} We use the product representations of the theta functions. Writing $q=e^{-\pi y}$ as usual, so $0<q<1$, we have
\begin{eqnarray*}
{\vartheta_4(y)} &=& {\prod_{m=1}^{\infty}(1-q^{2m})(1-q^{2m-1})^2} \\
\vartheta_3(y) &=& {\prod_{m=1}^{\infty}(1-q^{2m})(1+q^{2m-1})^2}
\end{eqnarray*}
Since the partial products $P_N = {\prod_{m=1}^{N}(1-q^{2m})(1-q^{2m-1})^2}$ and $Q_N = {\prod_{m=1}^{N}(1-q^{2m})(1+q^{2m-1})^2}$ converge to $\Th_4(y)$ and $\Th_3(y)$ respectively, and since $Q_N$ is clearly not zero for $0<q<1$, the quotient $P_N/Q_N$ converges to $\dfrac{\vartheta_4}{\vartheta_3}(y )$, and we have

\begin{eqnarray*}
\frac{\vartheta_4}{\vartheta_3}(y) &=& \lim_{N\rightarrow\infty} \frac{\prod_{m=1}^{N}(1-q^{2m})(1-q^{2m-1})^2}{\prod_{m=1}^{N}(1-q^{2m})(1+q^{2m-1})^2}\\ 
&=&\lim_{N\rightarrow\infty}\prod_{m=1}^{N}\left(\frac{1-q^{2m-1}}{1+q^{2m-1}}\right)^2=\lim_{N\rightarrow\infty}\prod_{m=1}^{N}\left(\frac{2}{1+q^{2m-1}}-1\right)^2 \\ 
&=& \prod_{m=1}^{\infty}\left(\frac{2}{1+q^{2m-1}}-1\right)^2. 
\end{eqnarray*}
Note that since $0<q<1$, 
\[
0<\left(\frac{2}{1+q^{2m-1}}-1\right)^2<1.
\]

As $y$ increases, $q$ strictly decreases, and $\left(\dfrac{2}{1+q^{2m-1}}-1\right)$ strictly increases. 
Hence, if $y > y'$, then the partial products (note that these start from $2$) 
$$R_N(y) = \prod_{m=2}^{N}\left(\dfrac{2}{1+q^{2m-1}}-1\right)^2$$
 satisfy $P_N(y) > P_N(y')$.  It follows that  $\lim_{N\rightarrow \infty} R_N(y) \ge \lim_{N\rightarrow \infty} R_N(y')$.  Note that $\lim_{N\rightarrow \infty} R_N(y) \neq 0$ for any $y$ with $0<y<1$  since $\Th_4(y)$ and $\Th_3(y)$ and hence $\dfrac{\vartheta_4}{\vartheta_3}(y)$ are nonzero for any $y$ with $0<y<1$.
Writing $q'$ for $e^{-\pi y'}$, we have for $m=1$ that $$\left(\dfrac{2}{1+q}-1\right)^2 > \left(\dfrac{2}{1+q'}-1\right)^2.$$ Hence we find  
\begin{eqnarray*}
\frac{\vartheta_4}{\vartheta_3}(y) &=& \left(\frac{2}{1+q}-1\right)^2 \prod_{m=2}^{\infty}\left(\frac{2}{1+q^{2m-1}}-1\right)^2 \\
& > & \left(\frac{2}{1+q'}-1\right)^2 \prod_{m=2}^{\infty}\left(\frac{2}{1+q'^{2m-1}}-1\right)^2 = \frac{\vartheta_4}{\vartheta_3}(y').
\end{eqnarray*}

Hence, $\dfrac{\vartheta_4}{\vartheta_3}(y)$ is a strictly increasing function of $y$.

As for the limits as $y$ tends to $0$ or $\infty$, note that $y\rightarrow 0$ precisely when $q\rightarrow 1$, and $y\rightarrow\infty$ precisely when $q\rightarrow 0$.
Now
\[
\frac{\vartheta_4}{\vartheta_3}(y)=\left(\frac{2}{1+q}-1\right)^2\prod_{m=2}^{\infty}\left(\frac{2}{1+q^{2m-1}}-1\right)^2\leq \left(\frac{2}{1+q}-1\right)^2,
\] and of course $\left(\dfrac{2}{1+q}-1\right)^2\rightarrow 0$
as $q\rightarrow 1$, i.e.,  when $y\rightarrow 0$.

Let us now consider the case $y\rightarrow\infty$, i.e.  $q\rightarrow 0$. Since $\vartheta_4(q)$ is absolutely convergent for $0<q<1$, we can group the terms in the following way: 
 \[
 \vartheta_4(q)=1-2q + 2(q^4-q^9) + 2(q^{16}-q^{25}) + \cdots > 1-2q,
 \] since $q^n > q^m$ when $n < m$.
 
On the other hand $q^4 - q^9 < q^4 < q$; $q^{16}-q^{25} < q^{16} < q^2$, etc., so
\[
\vartheta_4(q) < 1 -2q + 2(q + q^2 +\cdots)=1-2q + \frac{2q}{1-q}.
\]

Hence, $1-2q < \vartheta_4(q) < 1-2q + \frac{2q}{1-q}$. Now, as $q \rightarrow 0$, the two terms on either side of $\vartheta_4(q)$ tend to $1$, so $\vartheta_4(q)$ also tends to $1$.

We argue similarly for $\vartheta_3(q)$: $1 <\vartheta_3(q) < 1 +2(q + q^2  + q^3+ \cdots)$, where we have used $q^4 < q^2$, $q^9 < q^3$, $q^{16} < q^4$, etc.  Thus, we find $1 < \vartheta_3(q) < 1 + \frac{2q}{1-q}$.  Taking limits as $q\rightarrow 0$, we find $\vartheta_3(q)$ tends to $1$.

Thus, as $q\rightarrow 0$, $\vartheta_4(q)$ and $\vartheta_3(q)$ each tend to $1$, so their quotient tends to $1$.\end{proof}

\section {Proof of Lemma \ref{lem:val_alpha_at_peak}} \label{app:alpha_etc.}

\begin{proof} 
By \cite[Theorem 2.3]{B2}, for $k = \dfrac{\Th_2^2(q)}{\Th_3^2(q)}$,
\begin{equation} \label{B2Th2.3}
\pi \frac{K'(k)}{K(k)} = -\log q.
\end{equation}
Here, 
$$
K(k) = \int_0^1 \frac{dt}{(1-t^2)(1-k^2t^2)}
$$ and
$K'(k) = K(k')$, where $k' = \sqrt{1-k^2}$. By \cite[Exercise 4, \S 1.6]{B2},
\begin{equation} \label{KK'_ratio}
\dfrac{K'}{K}(\sqrt{2}-1) ={\sqrt{2}}. \end{equation}
In fact, we can see this as follows: Denoting $\sqrt{2}-1$ temporarily by $\alpha$, we have $\alpha' = \sqrt{1-\alpha^2} = \sqrt{2\alpha}$. By \cite[Theorem 1.2 (a), \S 1.4]{B2}, 
\begin{equation*}
K(\alpha) = \frac{1}{1+\alpha} K\left(\frac{2\sqrt{\alpha}}{1+\alpha}\right).
\end{equation*}
For our choice of $\alpha$, $\dfrac{2\sqrt{\alpha}}{1+\alpha}$ is just $\alpha'$, so the relation above becomes 
$K(\alpha) = \dfrac{1}{1+\alpha} K'(\alpha),$ which yields Equation \ref{KK'_ratio} above.

Now, since $k'=\dfrac{\Th_4^2(q)}{\Th_3^2(q)}$, we have in our situation $k'  = \sqrt{2}-1$, hence, for $k= \sqrt{1-k'^2}$, the expression $\dfrac{K'(k)}{K(k)}$ on the left side of Equation \ref{B2Th2.3}  equals $\dfrac{K(k')}{K'(k')} = \dfrac{K}{K'}(\sqrt{2}-1) = {1}/{\sqrt{2}}$. Hence, from Equation \ref{B2Th2.3} and the fact that $q = e^{-\pi y}$, we find $\pi  ({1}/{\sqrt{2}} )= \pi y$, so $y = {1}/{\sqrt{2}}$. Moroever, by Lemma \ref{kprime_increasing} above, this is the unique value of $y$ for which $\dfrac{\vartheta_4^2}{\vartheta_3^2}(y)$ attains this value.
\end{proof}
We note that theta functions have been computed for special values of $y$, some of which can be found in  \cite{Dieck}  for instance.

\end{document}